\documentclass[]{article}

\usepackage{authblk}

\usepackage[numbers]{natbib}

\usepackage[dvipsnames]{xcolor}
\usepackage[T1]{fontenc}
\usepackage[utf8]{inputenc}
\ifdefined\DeclareUnicodeCharacter
  \DeclareUnicodeCharacter{2011}{\mbox{-}}
\fi
\usepackage{enumerate}
\usepackage{dsfont}
\usepackage{amsfonts}
\usepackage{nicefrac}
\usepackage{amssymb}
\usepackage{amsthm}
\usepackage{mathtools}
\usepackage{tikz}
\usepackage{tabularx}
\usepackage{paralist}
\usepackage{booktabs}
\usepackage{hyperref}

\usepackage[capitalize]{cleveref}
\usepackage[backgroundcolor=gray!10,textsize=footnotesize]{todonotes}
\usepackage{xspace}
\usepackage{thm-restate}
\usepackage{orcidlink}

\usepackage{etoolbox}

\AtEndEnvironment{proof}{\let\qed\relax}

\newcommand{\appsymb}{$\star$}
\newcommand{\appref}[1]{{\hyperref[proof:#1]{\appsymb}}}

\usetikzlibrary{decorations.pathreplacing}
\usetikzlibrary{math}
\usetikzlibrary{calc}
\usetikzlibrary{shapes}

\newcommand{\problemdef}[3]{
	\begin{center}
	\begin{minipage}{0.95\columnwidth}
		\noindent
		\textsc{#1}
		\vspace{5pt}\\
		\setlength{\tabcolsep}{3pt}
		\begin{tabularx}{\textwidth}{@{}lX@{}}
			\textbf{Input:}     & #2 \\
			\textbf{Question:}  & #3
		\end{tabularx}
	\end{minipage}
	\end{center}
}

\newcommand{\XP}{\textrm{XP}\xspace}

\newcommand{\FPT}{\textrm{FPT}\xspace}
\newcommand{\NP}{\textrm{NP}\xspace}

\newcommand{\G}{\mathcal{G}}
\newcommand{\N}{\mathbb{N}}
\newcommand{\Z}{\mathbb{Z}}

\newcommand{\SEG}{{\normalfont\textsc{Seg-Agony}}\xspace}
\newcommand{\USEG}{{\normalfont\textsc{Unweighted Seg-Agony}}\xspace}

\DeclareMathOperator{\OPT}{OPT}

\def\tsc#1{\csdef{#1}{\textsc{\lowercase{#1}}\xspace}}
\tsc{WGM}
\tsc{QE}
\newtheorem{theorem}{Theorem}

\newtheorem{proposition}{Proposition}

\theoremstyle{definition}

\theoremstyle{plain}

\Crefname{theorem}{Theorem}{Theorems}
\crefname{theorem}{Thm.}{Thms.}
\Crefname{corollary}{Corollary}{Corollaries}
\crefname{corollary}{Cor.}{Cors.}

\newtheorem{claim}{Claim}

\newcommand{\sep}{\unskip;\space}
\newenvironment{keywords}
  {\par\smallskip\noindent\textbf{Keywords: }\ignorespaces}
  {\par\medskip}

\title{Parameterized Complexity of Temporal Agony}

\author{Tom-Lukas Breitkopf\,\orcidlink{0009-0008-2875-1945}, Vincent Froese\,\orcidlink{0000-0002-8499-0130}, Anton Herrmann\,\orcidlink{0009-0008-8473-9043}\thanks{Corresponding author},\\ Pascal Kunz\thanks{Supported by the DFG Research Training Group 2434 ``Facets of Complexity''}}
\affil{\parbox{0.95\textwidth}{\centering Technische Universität Berlin, Algorithmics and Computational Complexity Research Group, Germany}}
\affil{\texttt{t.breitkopf@tu-berlin.de}\\\texttt{vincent.froese@tu-berlin.de}\\\texttt{a.herrmann@tu-berlin.de}}
\date{}

\begin{document}

\maketitle

\begingroup
\renewcommand{\thefootnote}{}
\footnotetext{A conference version of this work was published at ALGOWIN~2026~\cite{conferenceVersion}.}
\endgroup

\begin{abstract}
Real-world networks are often organized in several layers forming a hierarchy which determines the interaction between the individual components. In order to discover such hierarchies in temporal networks, Tatti~[ECML PKDD 2018] introduced the \emph{temporal agony} problem \textsc{Seg-Agony}. Here, the goal is to assign each vertex a certain rank (from~1 to~$k$) such that arcs only point from lower ranks to higher ranks. Backward arcs are penalized depending on the difference between the corresponding ranks. Since arcs may change over time, each vertex is allowed to change its rank at most $\ell\ge 1$ times in order to minimize the overall penalty~$\alpha$ (called temporal agony).

We study the parameterized complexity of \SEG with a special focus on the number~$k$ of possible ranks for which we identify the precise complexity border. We show that the problem is polynomial-time solvable for~$k=2$, NP-hard for $k=3$ and~$\ell=1$ but fixed-parameter tractable for $\alpha$, and NP-hard for~$k=4$ and~$\ell=1$ even for~$\alpha=0$. We further show a polynomial-time algorithm for a constant number~$n$ of vertices and fixed-parameter tractability for the combined parameter~$n+\ell$.
\end{abstract}


\begingroup
\renewcommand{\thefootnote}{}
\endgroup

\section{Introduction}

Many real-world systems naturally exhibit a hierarchical structure between the considered objects or individuals.
For example, hierarchies appear in (wireless) computer networks, social networks, or also in biological networks (see references in \cite{Tatti17,Tatti18}). 
Detecting such hierarchies in directed networks is a key task for network analysis in various domains since it can explain the interaction between individuals and help to understand the observed outcome.

\citet{GSLMI11} introduced the following concept to discover hierarchies in directed networks:
For a given (static) digraph $G=(V,A,w)$ with positive arc weights~$w\colon A \to \N^+$, the goal is to find a \emph{rank assignment} $r\colon V\to [k]$ for some~$k\in\N$ such that the arcs obey these vertex ranks as best as possible, that is, all arcs should point from a lower rank to a higher rank.
Formally, given a \emph{penalty function} $p\colon \Z\to\N$ with~$p(x)=0$ for~$x<0$, we seek a rank assignment~$r$ that minimizes
\[q(r,G,p)\coloneqq \sum_{(u,v)\in A}w(u,v)\cdot p(r(u)-r(v)).\]
If the penalty~$p$ is the function $p_l(x)\coloneqq\max(0,x+1)$, then the cost~$q$ above is called (weighted) \emph{agony}~\cite{GSLMI11}.
Note that for~$k=|V|$ and ``constant'' penalty~$p_c(x)=\mathds{1}_{x\ge 0}$ the problem of minimizing~$q(r,G,p_c)$ for unweighted digraphs is equivalent to the well-known \NP-hard \textsc{Feedback Arc Set} problem.

Since many real-world networks are in fact dynamically changing over time,
\citet{Tatti18} extended this approach to temporal (weighted) directed networks
where the arcs may appear/disappear over discrete time steps and have changing weights.
We focus on one of their problem formulations where the vertex rank is allowed to change at most ~$\ell\ge 1$ times over the~$\tau$ time steps.
To this end, a \emph{temporal rank assignment} $r\colon V\times [\tau]\to[k]$ is called a \emph{rank $\ell$-segmentation} if there exist rank functions~$r_0,\ldots,r_\ell\colon V\to[k]$ and a change point function~$c\colon V\times[\ell]\to[\tau]$ with $c(v,i+1)\ge c(v,i)$ for all~$v\in V$ and~$i\in[\ell-1]$ such that
\[
r(v,t)=
\begin{cases}r_0(v),& \text{ if } t<c(v,1),\\
r_1(v),& \text{ if } c(v,1)\le t < c(v,2),\\
\vdots\\
r_\ell(v),& \text{ if } c(v,\ell)\le t,
\end{cases}
\]
holds for all vertices~$v\in V$ and time steps~$t\in[\tau]$.
For a temporal weighted digraph $\G=(V,(A_t)_{t\in[\tau]},(w_t)_{t\in[\tau]})$ with temporal arc weights $w_t\colon A_t\to\N^+$, the goal is
then to find a rank $\ell$-segmentation~$r$ that minimizes the following cost (called \emph{temporal agony})
\[q(r,\G,p_l)\coloneqq \sum_{t\in[\tau]}\sum_{(u,v)\in A_t}w_t(u,v)\cdot p_l(r(u,t)-r(v,t)).\]

The decision version is formally defined as follows:
\problemdef{\SEG}
{A temporal weighted digraph~$\G=(V,(A_t)_{t\in[\tau]},(w_t)_{t\in[\tau]})$, two integers $k$, $\ell\in\N$, and~$\alpha\in\N$.}
{Is there a rank $\ell$-segmentation~$r$ with~$q(r,\G,p_l)\le \alpha$?}

Note that if no rank changes are allowed (that is, $\ell=0$), then the problem reduces to the static case (by taking the union of all arc sets over time and summing up the corresponding weights).

In this work, we initiate a more fine-grained complexity analysis of \SEG in terms of the parameterized complexity in order to delineate the exact border of computational tractability.
Our results reveal interesting connections to Boolean satisfiability problems with transitions from polynomial-time solvable cases to NP-hardness for constant parameter values (involving the parameters number~$k$ of ranks, the target agony~$\alpha$, the number~$\ell$ of change points, and the lifetime~$\tau$).

\paragraph*{Related Work.}
\citet{GSLMI11} showed that for a static unweighted digraph with~$n$ vertices and~$m$ arcs agony can be minimized in~$O(nm^2)$ time if~$k=n$.
\citet{Tatti17} improved the running time to~$O(m^2)$ and also obtained a running time of $O(m(\min(nk,m)+n\log n))$ for arbitrary~$k$
by showing that the problem is the dual of a capacitated circulation problem.
Moreover, the algorithm also solves the problem with integer weights with a running time increasing by a factor of~$O(\log W)$, where~$W$ is the maximum arc weight.
\citet{Tatti17} further showed polynomial-time solvability for every convex penalty function~$p$ and \NP-hardness for a concave penalty (and~$k=4$) and developed several fast heuristics.

For temporal digraphs, \citet{Tatti18} proved \NP-hardness of~\SEG and developed a polynomial-time heuristic for~$\ell=1$ change point.
The \NP-hardness \cite[Proposition~2]{Tatti18} already holds for unweighted temporal digraphs with~$\tau=3$, $\ell=1$ and~$\alpha=0$. The proof\footnote{See arXiv version \url{https://arxiv.org/pdf/1902.01873.pdf}.}, however, does not explicitly specify the number~$k$ of allowed ranks in the constructed instance. A closer inspection shows that the reduction requires~$k\ge 8$ ranks.
\citet{Tatti18} also introduced a second problem called \textsc{Fluc-Agony} where the \emph{fluctuation}, that is, the total sum of rank differences of each vertex over consecutive time steps is minimized. \textsc{Fluc-Agony} is solvable in polynomial time~\cite[Section~4]{Tatti18}.

\paragraph*{Our Contributions.}
We study the parameterized complexity of \SEG with respect to the parameters~$n$ (number of vertices), $k$ (number of ranks), $\ell$ (number of change points), $\alpha$ (target temporal agony), and~$\tau$ (lifetime). Note that the problem is already known to be NP-hard even for~$\ell=1$, $\alpha=0$, $\tau=3$, and~$k\ge 8$ \cite{Tatti18}.
Hence, it cannot be in \XP for the combination of all these four parameters. Clearly, the case~$\ell=0$ or~$\tau \le 2$ is polynomial-time solvable.
We therefore focus on the number~$k$ of ranks (with $\ell=1$) and completely classify its computational complexity. Note that it is natural to assume that~$k$ is rather small since in practice reasonable hierarchies should not have too many levels.
We show that \SEG is polynomial-time solvable for~$k=2$ (\Cref{thm:k=2}) and \NP-hard for~$k=3$ (\Cref{thm:three_ranks}) but fixed-parameter tractable with respect to~$\alpha$ (\Cref{thm:three_ranks_FPT}).
For $k=4$, however, we even exclude containment in \XP (\Cref{thm:four_ranks}) by showing that the problem is \NP-hard even with constant~$\alpha$, $\ell$ and~$\tau$ (strengthening the previous NP-hardness by \citet{Tatti18}).
On the positive side, we also show that the problem is in \FPT with respect to the combined parameter~$n+\ell$ and in \XP for the parameter~$n$ alone (\Cref{thm:fpt-n-ell}).

\section{Preliminaries}
We start with some basic definitions. For~$n\in\N$, let~$[n]\coloneqq \{1,\ldots,n\}$ and for~$n\le n'\in\N$ let~$[n,n']\coloneqq \{n,n+1,\ldots,n'\}$.

\paragraph*{Temporal Digraphs.}
A \emph{directed} graph (\emph{digraph}) is a pair~$G=(V,A)$ with vertex set~$V$ and arc set~$A\subseteq V^2$.
A \emph{temporal digraph} $\G=(V,(A_t)_{t\in[\tau]})$ with~\emph{lifetime}~$\tau\in\N$ consists of a finite vertex set $V$ and a sequence $A_1,\ldots,A_\tau \subseteq V^2$ of arc sets.
An arc of the form~$(u,v)$ is an \emph{incoming arc} of~$v$ and an arc of the form~$(v,u)$ is an \emph{outgoing arc} of~$v$.
The static digraph~$G_t\coloneqq (V,A_t)$ is called the~\emph{$t$-th layer} of $\G$.
The \emph{size} of $\G$ is $|\G| \coloneqq n + A_\G$ where~$n\coloneqq|V|$ and $A_\G\coloneqq\sum_{t=1}^\tau \max\{1,|A_t|\}$.
The \emph{underlying} digraph of~$\G$ is the static digraph~$(V,\bigcup_{t\in[\tau]}A_t)$.
A \emph{DAG} is a directed acyclic graph.

\paragraph*{Parameterized Complexity.}
A \emph{parameterized problem} is a language $L \subseteq \Sigma^* \times \mathbb{N}$, where $\Sigma$ is a fixed, finite alphabet.
For an instance $(x,k) \in \Sigma^* \times \mathbb{N}$, the number $k$ is called the \emph{parameter}.
A parameterized problem is contained in the complexity class XP if it can be solved in~$|x|^{f(k)}$ time for some computable function~$f$.
Moreover, it lies in the subclass FPT if it can be solved in~$f(k)\cdot |x|^{O(1)}$ time for some computable function~$f$.
For further details on parameterized complexity we refer the reader to the book by~\citet{CyganFKLMPPS15}.

\section{Parameter~$n + \ell$}
In this section we show that \SEG is in FPT with respect to the combined parameter~$n + \ell$ and also in XP for~$n$ alone.

The basic idea is to use dynamic programming over all possible rank $\ell$-segmentations.
Since every vertex has at most~$\ell + 1$ different ranks in a rank $\ell$-segmentation, it is clear that we can bound the number~$k$ of ranks from above by~$(\ell + 1)n$.
This observation is used in the dynamic program in the proof of the following theorem.

\begin{restatable}{theorem}{fptnell}
	\label{thm:fpt-n-ell}
	\SEG is solvable in~$\tau (\ell n)^{O(n)}\log W$ time, where~$W$ is the maximum weight.
\end{restatable}

\begin{proof}
  We use dynamic programming with a table~$T$ containing an entry for every~$t\in[\tau], c\colon V \to \{0,\ldots,\ell\}$ and~$\rho\colon V \to [k]$.
  The entry $T[t,c,\rho]$ is defined as the minimum temporal agony until time step~$t$ given that~$v \in V$ changed its rank~$c(v)$ many times and has rank~$\rho(v)$ in step~$t$ (we say a rank $\ell$-segmentation~$r$ satisfying this \emph{respects}~$c$ and~$\rho$).
  An entry is undefined if~$c(v)\geq t$ for any~$v \in V$.
  We say that a pair~$(c',\rho') \in (V\to \{0,\ldots,\ell\}) \times (V \to [k])$ is \emph{compatible} with the pair~$(c,\rho) \in (V \to \{0,\ldots,\ell\}) \times (V \to [k])$ if for all~$v \in V$ it holds that~$c'(v) \in \{c(v)-1,c(v)\}\cap\N$ and~$c'(v) = c(v) \iff \rho'(v) = \rho(v)$.
  We denote this with~$(c',\rho') \preceq (c,\rho)$.

  The table~$T$ is computed as follows:
  \begin{align*}
   T[1,\mathbf{0},\rho] &= q(\rho,(V,A_1,w_1),p_l)\\
   T[t,c,\rho] &= \min\limits_{(c',\rho')\preceq(c,\rho)}T[t-1,c',\rho'] + q(\rho,(V,A_t,w_t),p_l) \qquad\text{(for }t>1)
  \end{align*}
\noindent
  The minimum agony over all rank $\ell$-segmentations can then be obtained via \[\min\limits_{\substack{c\colon V \to \{0,\ldots,\ell\}\\ \rho\colon V\to [k]}}T[\tau,c,\rho].\]

  As regards the running time, note that there are~$\tau$ time steps,~$(\ell+1)^n$ functions~$c$ and we can assume~$k^n\leq ((\ell+1)n)^n$ ranking functions~$\rho$ (using the previously made observation about the number of ranks).
  This yields an overall table size of~$\tau(\ell+1)^n(\ell+1)^n n^n \in \tau(\ell n)^{O(n)}$.
  For every entry, at most $k^n\in(\ell n)^{O(n)}$ other entries need to be considered (given~$\rho$, when iterating over all~$\rho'$ the function~$c'$ is implied).
  Computing the agony of a static graph can be done in~$n^{O(1)}\log W$ time~\cite{Tatti17}.
  This gives an overall running time of~$\tau (\ell n)^{O(n)}\log W$.

  We prove correctness via induction over the time steps~$t$.
  For~$t \in [\tau]$, we define~$\mathcal{G}_t\coloneqq (V,(A_{t'})_{t' \in [t]},(w_{t'})_{t'\in[t]})$.
  With~$\OPT(t,c,\rho)$ we denote the minimum possible temporal agony for~$\mathcal{G}_t$ respecting~$c$ and~$\rho$.
  For~$t=1$, clearly~$T[1,\mathbf{0},\rho]=\OPT(1,\mathbf{0},\rho)$ as we are in the static case.
  Given that~$T[t-1,c',\rho'] = \OPT(t-1,c',\rho')$ for any~$c',\rho'$, a feasible solution respecting~$c$ and~$\rho$ with temporal agony~$T[t-1,c',\rho'] + q(\rho,(V,A_t,w_t),p_l)$ can be constructed in case~$(c',\rho')$ is compatible with~$(c,\rho)$ by changing the rank of any vertex~$v$ from~$\rho'(v)$ to~$\rho(v)$ if~$c'(v) = c(v)-1$.
  Thus,~$\OPT(t,c,\rho) \leq T[t,c,\rho]$.
  Since we consider all~$(c',\rho')$ compatible with~$(c,\rho)$, it also holds~$T[t,c,\rho] \leq \OPT(t,c,\rho)$ as the optimal rank $\ell$-segmentation respecting~$c$ and~$\rho$ must be compatible with some~$(c',\rho')$ in step~$t-1$.\qed
\end{proof}

\section{Two Ranks}
In this section we show that \SEG is polynomial-time solvable for~$k = 2$ ranks.
The idea is to show first that (for any penalty function~$p$ with~$p(1)=2p(0)$) the temporal agony~$q(r,\G,p)$ is essentially determined by the weights of in- and outgoing arcs of each vertex (and does not depend on the ranks of its neighbors).

\begin{proposition}
	\label{obs:k=2}
	Let $r\colon V \times [\tau] \to [2]$ be a rank $\ell$-segmentation for the temporal weighted digraph $\G$ and let~$p$ be a penalty
        function with~$p(1)=2p(0)$.
	Then, \[ q(r,\G,p) = p(0)\sum_{t =1}^\tau \left( \sum_{\substack{v \in V\\r(v,t) =1}}\sum_{(u,v)\in A_t}w_t(u,v) +  \sum_{\substack{v \in V\\r(v,t) =2}} \sum_{(v,u)\in A_t}w_t(v,u)\right).\]
\end{proposition}

\begin{proof}
  By definition, we have \[q(r,\G,p)=\sum_{t=1}^\tau\sum_{(u,v)\in A_t}w_t(u,v)p(r(u,t)-r(v,t)).\]
  For the inner sum, it holds
  \begin{align*}
    &\sum_{(u,v)\in A_t}w_t(u,v)p(r(u,t)-r(v,t)) = \sum_{\mathclap{\substack{(u,v)\in A_t\\r(u,t)=r(v,t)}}}w_t(u,v)p(0) + \sum_{\mathclap{\substack{(u,v)\in A_t\\r(u,t)=2,r(v,t)=1}}}w_t(u,v)2p(0)\\
    &= p(0)\left(\sum_{\substack{v\in V\\r(v,t)=1}}\sum_{\substack{(u,v)\in A_t\\r(u,t)=1}}w_t(u,v) + \sum_{\substack{v\in V\\r(v,t)=2}}\sum_{\substack{(v,u)\in A_t\\r(u,t)=2}}w_t(v,u) + 2\sum_{\mathclap{\substack{(u,v)\in A_t\\r(u,t)=2,r(v,t)=1}}}w_t(u,v)\right)\\
    &= p(0)\left(\sum_{\substack{v \in V\\r(v,t) =1}} \sum_{(u,v)\in A_t}w_t(u,v) +  \sum_{\substack{v \in V\\r(v,t) =2}} \sum_{(v,u)\in A_t}w_t(v,u)\right).
  \end{align*}\qed
\end{proof}

With \Cref{obs:k=2}, we are now ready to state the algorithm for \SEG.

\begin{theorem}\label{thm:k=2}
	\SEG is solvable in~$O(n^2\tau\ell\log(n\tau W))$ time for~$k=2$, where~$W$ is the maximum arc weight.
\end{theorem}

\begin{proof}
	By \Cref{obs:k=2}, the optimal rank of a vertex $v\in V$ does not depend on the ranks of the other vertices.
	We give a dynamic program that computes the optimal rank for~$v$.
	We define a table $T_v$ with an entry $T_v[t,i,j]$ for each $t\in[\tau]$, $i\in[2]$, and $j \in \{0,\ldots,\ell\}$.
	The entry $T_v[t,i,j]$ is the minimum contribution of $v$ to the agony in the first $t$ layers of $\G$ for any rank $j$-segmentation~$r$ with $r(v,t) = i$.

        By \Cref{obs:k=2}, it holds $T_v[1,1,j]=\sum_{(u,v)\in A_1}w_1(u,v)$ and $T_v[1,2,j]=\sum_{(v,u)\in A_1}w_1(v,u)$ for all $j$.
	For $t>1$, it holds
        \begin{align*}
          T_v[t,1,0] & = T_v[t-1,1,0]+\sum_{(u,v)\in A_t}w_t(u,v),\\
          T_v[t,2,0] & = T_v[t-1,2,0]+\sum_{(v,u)\in A_t}w_t(v,u)
        \end{align*}
        and for~$j>0$ it holds
	\begin{align*}
          T_v[t,1,j] & = \min (T_v[t-1,1,j], T_v[t-1,2,j-1]) + \sum_{(u,v)\in A_t}w_t(u,v),\\
          T_v[t,2,j] & = \min (T_v[t-1,2,j], T_v[t-1,1,j-1]) + \sum_{(v,u)\in A_t}w_t(v,u).       
	\end{align*}
	The minimum cost of any rank $\ell$-segmentation is \[\sum_{v\in V} \min (T_v[\tau,1,\ell],T_v[\tau,2,\ell]).\]
	By adding an appropriate traceback procedure, we can also obtain the optimal rank assignment.

	For every $v\in V$, the table $T_v$ has~$O(\tau\ell)$ entries and each of them can be computed in~$O(n\log(n\tau W))$ time.
	Here, the factor~$\log(n\tau W)$ comes from the fact that the maximum value in a table entry is bounded by~$n\tau W$.
	Thus, the overall running time is~$O(n^2\tau\ell\log(n\tau W))$.\qed
\end{proof}

      As a side result, we strengthen the NP-hardness result by~\citet[Proposition~4]{Tatti17} for the static case which states that it is \NP-hard to minimize~$q(r,G,p)$ for certain concave penalties~$p$ (the reduction uses $k=4$ ranks and weighted arcs).
We show NP-hardness for $k=2$ on unweighted static digraphs.

\begin{restatable}{theorem}{kTwoStatic}
	\label{thm:k=2_static}
	Minimizing~$q(r,G,p)$ with~$k=2$ ranks is \NP-hard on static unweighted digraphs~$G$ for every penalty function~$p$ with~$p(1) < 2p(0)$.
  If~$0<p(0)=p(1)$, then it is even \NP-hard on unweighted DAGs.
\end{restatable}

\begin{proof}
  We give a polynomial-time reduction from \textsc{Maximum Cut} \cite{GareyJ79}. The problem is to decide whether the vertices of an undirected graph~$H=(V,E)$ can be partitioned into two subsets~$V_1$ and~$V_2$ such that
  $|\{\{u,v\}\in E\mid u\in V_1, v\in V_2\}|\ge k$. For an instance~$(H,k)$, we build the digraph~$G=(V,A)$ with~$A\coloneqq\{(u,v),(v,u)\mid \{u,v\}\in E\}$. We claim that~$H$ has a cut of size at least~$k$ if and only if there exists a rank assignment~$r\colon V \to [2]$ with~$q(r,G,p)\le 2p(0)|E|-k(2p(0)-p(1))$.
  To verify the correctness, assume first that~$H$ has a cut of size $s\ge k$ between the vertex sets~$V_1$ and~$V_2$. Then, let~$r(v)\coloneqq 1$ for all~$v\in V_1$ and~$r(v)\coloneqq 2$ for all~$v\in V_2$ and note that
  \begin{align*}
    q(r,G,p)&=2p(0)(|E|-s) +p(1)s=2p(0)|E|-s(2p(0)-p(1))\\
    &\le 2p(0)|E|-k(2p(0)-p(1))
  \end{align*}
  since~$s\ge k$ and~$2p(0)>p(1)$.

  Conversely, let~$r$ be a rank assignment with~$q(r,G,p)\le 2p(0)|E|-k(2p(0)-p(1))$ and let~$V_1\coloneqq \{v\mid r(v)=1\}$, $V_2\coloneqq\{v\mid r(v)=2\}$, and let~$s$ be the size of the cut between~$V_1$ and~$V_2$. Then, we have
  \begin{align*}
    q(r,G,p)= 2p(0)|E|-s(2p(0)-p(1))\le 2p(0)|E| - k(2p(0)-p(1)),
  \end{align*}
  which implies~$s\ge k$ since~$2p(0)>p(1)$.
  
  For the \NP-hardness on DAGs, we reduce from \textsc{Maximum Directed Cut}, which is \NP-hard on DAGs~\cite{LKM11}. For a given DAG~$G=(V,A)$ and integer~$k$, the task is to determine whether there exists a subset~$S\subseteq V$ such that $|\{(u,v)\in A\mid u\in S, v\in V\setminus S\}|\ge k$.
  We claim that~$G$ has a directed cut of size at least~$k$ if and only if there exists a rank assignment~$r\colon V\to[2]$ with~$q(r,G,p)\le p(1)(|A|-k)$.

  Let~$S\subseteq V$ be a subset with~$s\ge k$ outgoing arcs. Define~$r(v)\coloneqq 1$ for~$v\in S$ and~$r(v)\coloneqq 2$ for~$v\in V\setminus S$.
  Then $q(r,G,p)\le p(1)(|A|-s)\le p(1)(|A|-k)$ since~$p(0)=p(1)>0$.

  Let now~$r$ be a rank assignment with~$q(r,G,p)\le p(1)(|A|-k)$ and let~$S\coloneqq \{v\in V\mid r(v)=1\}$ and~$s$ be the number of outgoing arcs from~$S$.
  Then, $q(r,G,p)=p(1)(|A|-s)$ since~$p(1)=p(0)$, and thus~$s \ge k$ since~$p(1)>0$.\qed
\end{proof}

\section{Three Ranks}
\label{sec:three_ranks}

In this section, we study \SEG restricted to $k = 3$ ranks and~$\ell = 1$ change points.
We show that the problem is NP-hard in general but fixed-parameter tractable with respect to the temporal agony~$\alpha$.
Interestingly, the NP-hardness reduction is from~\textsc{Max 2-SAT} while one of the key steps in the algorithm is a~\textsc{Max 2-SAT} encoding.

We start by providing the hardness result.

\begin{theorem}\label{thm:three_ranks}
	\USEG is NP-hard even for $k=3$ and $\ell =1$.
\end{theorem}

\begin{proof}
 We give a polynomial-time reduction from \textsc{Max 2-SAT} where the input is a~$2$-CNF formula~$\Phi$ and an integer~$d$ and the question is whether there exists an assignment in which at most~$d$ clauses are not satisfied.
 We assume that each variable appears exactly three times positively and three times negatively \cite{BermanK99}.
 Suppose~$\Phi$ has the variable set~$\{x_1, \dots, x_n\}$ and clause set~$\{C_1, \dots, C_m\}$ where~$C_j = (c_{j_1} \lor c_{j_2})$.

 \noindent
 \emph{Construction:} For an intuition and a sketch of the following construction see~\Cref{fig:three_ranks_np_hardness}.
 \begin{figure}[t]
	\centering
	\begin{tikzpicture}[scale = 0.504, transform shape,
		V/.style = {circle, draw, fill=black}
		]
		\node (nl1) at (-13,0.75) {\LARGE{neg. literal}};
		\node (l1) at (-13,3.75) {\LARGE{pos. literal}};
		\node (A1) at (-10,0.75) {\LARGE$a_{i_1}$};
		\node (M1) at (-12, 2.25) {\LARGE$m_{i_2}$};
		\node (V1) at (-10, 2.25) {\LARGE$v_{j}$};
		\node (nV1) at (-8, 2.25) {\LARGE$\overline{v}_j$};
		\node (B1) at (-10, 3.75) {\LARGE$b_{i_3}$};
		\node (T1) at (-10, -1) {\huge $t = 1$};
		\draw[->,thick] (A1) edge (M1);
		\draw[->,thick] (M1) edge (l1);
		\draw[->,thick] (nl1) edge (M1);
		\draw[->,thick] (A1) edge (V1);
		\draw[->,thick] (A1) edge (nV1);
		\draw[->,thick] (M1) edge (B1);
		\draw[->,thick] (V1) edge (B1);
		\draw[->,thick] (nV1) edge (B1);

		\node (V2) at (-4, 3) {\LARGE$v_{j}$};
		\node (nV2) at (-4, 1.5) {\LARGE$\overline{v}_j$};
		\node (A2) at (-4, -1) {\LARGE $t \in [2, \alpha + 2]$};
		\draw[->,thick] (nV2) edge (V2);

		\node (nl3) at (0,0.75) {\LARGE{neg. literal}};
		\node (V3) at (0, 2.25) {\LARGE$v_{j}$};
		\node (B3) at (0, 3.75) {\LARGE$b_{i}$};
		\node (nV3) at (3, 2.25) {\LARGE$\overline{v}_j$};
		\node (l3) at (3,3.75) {\LARGE{pos. literal}};
		\node (A3) at (3, 0.75) {\LARGE$a_{i}$};
		\node (T1) at (1.5, -1) {\LARGE $t \in [\alpha + 3, \alpha + d +  3]$};
		\draw[->,thick, bend left] (nl3) edge (V3);
		\draw[->,thick, bend left] (V3) edge (nl3);
		\draw[->,thick, out = 35, in = -35] (nl3) edge (B3);
		\draw[->,thick, bend left] (l3) edge (nV3);
		\draw[->,thick, bend left] (nV3) edge (l3);
		\draw[->,thick, out = 35, in = -35] (A3) edge (l3);

		\node (A4) at (7,0.75) {\LARGE$a_{i_1}$};
		\node (M4) at (7, 2.25) {\LARGE$m_{i_2}$};
		\node (B4) at (7, 3.75) {\LARGE$b_{i_3}$};
		\node (c41) at (9, 1.5) {\LARGE$c_{j_1}$};
		\node (c42) at (9, 3.) {\LARGE$c_{j_2}$};
		\node (T1) at (8, -1) {\LARGE $t = \alpha + d + 4$};
		\draw[->,thick] (A4) edge (M4);
		\draw[->,thick] (M4) edge (B4);
		\draw[->,thick] (c41) edge (c42);

		\draw[dashed] (-6, 0) -- (-6, 4);
		\draw[dashed] (-2, 0) -- (-2, 4);
		\draw[dashed] (5, 0) -- (5, 4);

	\end{tikzpicture}
	\caption{Given a \textsc{Max 2-SAT} formula where we want to satisfy all but at most~$d$ clauses, the reduction in \Cref{thm:three_ranks} constructs a temporal graph with~$\alpha + d + 4$ time steps.
	The figure sketches the arc sets at different time steps~$t$.
	The first time step fixes a ``starting rank'' for every introduced vertex.
	The time steps~$[2, \alpha + 2]$ ensure that either~$v_j$ switches its rank from 2 to 3 or~$\overline{v}_j$ switches from rank 2 to 1.
	This corresponds to an assignment of the variable~$x_j$ where moving~$v_j$ to rank 3 corresponds to setting the variable to true.
	The time steps~$[\alpha + 3, \alpha + d + 3]$ force the negative or positive literals of a variable~$x_j$ to move to rank 2 depending on the earlier decision.
	Finally, the last time step ensures that~$a_i,m_i$ and~$b_i$ do not change their rank at any time step and the arc~$(c_{j_1}, c_{j_2})$ causes agony if and only if both of the literals were moved to rank 2 previously (which means both of them were set to false).}
	\label{fig:three_ranks_np_hardness}
\end{figure}
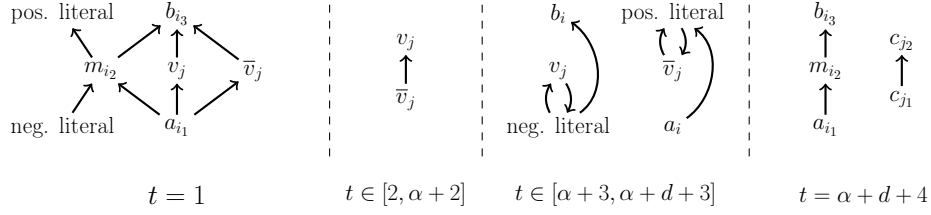

 Let~$\alpha \coloneqq 12n(d+1) + d$.
 We construct a temporal digraph~$\G=(V,(A_t)_{t\in[\tau]})$ with~$\tau \coloneqq \alpha + d + 4$ time steps.
 The vertex set~$V$ is defined as follows:
 \begin{itemize}
  \item For each variable~$x_i$, we add two vertices~$v_i$ and~$\overline{v}_i$ and refer to them as \emph{variable vertices}.
  \item For each literal~$c_{j_z}$, we add one vertex~$c_{j_z}$ and refer to it as a \emph{negative/positive literal vertex}.
  \item For each~$i \in [\alpha + 1]$, we add three vertices~$a_i, b_i$ and~$m_i$.
 \end{itemize}

 For~$t \in [\tau]$, we define the arc set~$A_t$ as follows:
 \begin{itemize}
  \item $t=1$: For all~$i_1,i_2 \in [\alpha + 1]$, we add the arcs~$(a_{i_1}, m_{i_2})$ and~$(m_{i_1}, b_{i_2})$.
  For every variable~$x_j$ and every~$i \in [\alpha + 1]$, we add the arcs~$(a_i, v_j)$, $(v_j, b_i)$, $(a_i, \overline{v}_j)$ and~$(\overline{v}_j, b_i)$.
  Moreover, for every literal vertex~$c_{j_z}$ and every~$i \in [\alpha + 1]$, we add the arc~$(c_{j_z}, m_i)$ if~$c_{j_z}$ represents a negative literal, and if~$c_{j_z}$ represents a positive literal, then we add~$(m_i, c_{j_z})$.
  \item $t \in [2, \alpha  + 2]$: For every variable~$x_i$, we add the arc~$(\overline{v}_i, v_i)$.
  \item $t \in [\alpha + 3, \alpha + d + 3]$: For every literal vertex~$c_{j_z}$, we add the following arcs.
  If~$c_{j_z}$ represents a positive literal of the variable~$x_i$, then we add the arcs~$(c_{j_z}, \overline{v}_i), (\overline{v}_i,c_{j_z})$ and for every~$i \in [\alpha +1]$ the arc~$(a_i, c_{j_z})$.
  If~$c_{j_z}$ represents a negative literal of the variable~$x_i$, then we add the arcs~$(c_{j_z}, v_i), (v_i,c_{j_z})$ and for every~$i \in [\alpha +1]$ the arc~$(c_{j_z}, b_i)$.
  \item $t = \alpha + d + 4$: For every clause~$(c_{j_1} \lor c_{j_2})$, we add the arc~$(c_{j_1}, c_{j_2})$ and for all~$i_1,i_2 \in [\alpha + 1]$, we add the arcs~$(a_{i_1}, m_{i_2})$ and~$(m_{i_1}, b_{i_2})$.
 \end{itemize}

 \emph{Correctness:} We proceed by showing that there is a variable assignment for~$\Phi$ which satisfies at least~$m-d$ clauses if and only if there is a rank 1-segmentation of~$\G$ with temporal agony at most~$\alpha$.

 ``$\Rightarrow:$'' Let~$\beta$ be a variable assignment which satisfies at least~$m-d$ clauses of $\Phi$.
 Consider the following rank 1-segmentation~$r$:
 \begin{itemize}
  \item For all~$i \in [\alpha + 1]$ and~$t\in [\tau]$, we define~$r(a_i,t) \coloneqq 1$, $r(m_i,t) \coloneqq 2$ and~$r(b_i,t) \coloneqq 3$.
  \item If~$\beta(x_i) = 1$, then we define
  \begin{align*}
  r(v_i,t)\coloneqq \begin{cases}
                      2, &t = 1\\
                      3, &t > 1 \\
                    \end{cases}
    \ \text{ and } \
  r(\overline{v}_i,t)\coloneqq 2 \text{ for all }t\in[\tau],
  \end{align*}
  and if~$\beta(x_i) = 0$, then we define
  \begin{align*}
  r(v_i,t)\coloneqq 2 \text{ for all }t\in[\tau]
  \ \text{ and } \
  r(\overline{v}_i,t)\coloneqq \begin{cases}
                      2, &t = 1\\
                      1, &t > 1 \\
                    \end{cases}.
  \end{align*}
  \item For every negative literal vertex~$c_{j_z}$, we define~$r(c_{j_z},1) \coloneqq 1$ and for every positive literal vertex we set~$r(c_{j_z}, 1) \coloneqq 3$.
  If~$c_{j_z}$ evaluates to~0 under~$\beta$, then we set~$r(c_{j_z},t) \coloneqq 2$ for every~$t > 1$.
  If~$c_{j_z}$ evaluates to~1, then we set~$r(c_{j_z}, t) \coloneqq r(c_{j_z},1)$ for~$t \in \{2,\dots,\tau-1\}$ and define~$r(c_{j_z}, \tau)$ such that $r(c_{j_1},\tau)<r(c_{j_2},\tau)$.
  Note that this is possible.
  If both literals are true, then we can set~$r(c_{j_1},\tau)\coloneqq 1$ and~$r(c_{j_2},\tau)\coloneqq 2$.
  If one literal of the clause is false, then it has rank 2 and consequently we can set the rank of the literal vertex corresponding to the true literal of the clause to~$1$ or~$3$ depending on whether it is a source or sink in~$(V, A_\tau)$.
 \end{itemize}
 First, note that all vertices change their rank at most once.
 It remains to check that~$q(r,\G,p_l)\le \alpha$.
 For this, observe that in the first~$\alpha + 2$ time steps all arcs point from a smaller rank to a larger rank.
 In the time interval~$[\alpha + 3, \alpha + d + 3]$, the only arcs that point from a larger rank to a smaller rank are arcs between variable vertices and corresponding literal vertices.
 More precisely, each vertex~$v_i$ has an outgoing arc to the three negative literal vertices of the variable~$x_i$.
 By construction, these literal vertices are one rank below~$v_i$ which yields a penalty of~6.
 Similarly, the vertex~$\overline{v}_i$ has three incoming arcs from the positive literal vertices which are ranked one rank above~$\overline{v_i}$ and also yield a penalty of~$6$.
 Since the time interval has~$d+1$ time steps and there are~$n$ pairs of variable vertices, the penalties sum up to~$12n(d+1)$.
 Finally, note that in the last time step only the arcs corresponding to unsatisfied clauses yield a penalty of~1 since the two corresponding literal vertices have rank~$2$.
 Hence, the temporal agony is at most~$12n(d+1) + d = \alpha$.

  ``$\Leftarrow:$'' Let~$r$ be a rank 1-segmentation with~$q(r,\G,p_l)\le\alpha$.
  We start with some basic observations about~$r$.
  For every~$i \in [\alpha + 1]$, each vertex~$a_i,m_i$ and~$b_i$ is contained in at least~$\alpha + 1$ arc-disjoint directed paths of length two in the first and last time step.
  In particular,~$a_i$ is always a start vertex, $b_i$ is always an end vertex and~$m_i$ is always in the middle.
  Hence, we conclude~$r(a_i, t) = 1, r(m_i,t) = 2$ and~$r(b_i,t) = 3$ for all~$t \in [\tau]$, as otherwise the agony would already be larger than~$\alpha$.
  Next, we consider the different time steps and discuss the behavior of~$r$.

  $t = 1$: Each variable vertex has an incoming arc from every~$a_i$ and outgoing arc to every~$b_i$.
  Since there are~$\alpha +1$ many~$a_i$'s and~$b_i$'s and they have rank 1 and 3 respectively, we conclude that~$r(v_i , 1) = r(\overline{v}_i,1) = 2$.
  If~$c_{j_z}$ is a negative literal vertex, then~$(c_{j_z},m_i) \in A_1$ for all~$i \in [\alpha +1]$ and since $r(m_i,1)=2$, this forces~$r(c_{j_z},1) = 1$.
  By the same argument, we derive~$r(c_{j_z},1) = 3$ for every positive literal vertex~$c_{j_z}$.

  $t \in [2, \alpha  + 2]$: The only arcs in~$A_t$ are of the form~$(\overline{v}_i,v_i)$, so we can assume that these are the only vertices which possibly change their rank in this time interval.
  Since the arc~$(\overline{v}_i,v_i)$ appears in~$\alpha + 1$ consecutive time steps and~$r(v_i , 1) = r(\overline{v}_i,1) = 2$, it is clear that~$v_i$ or~$\overline{v}_i$ has to change its rank.
  We can assume that this rank change is at time step 2 and that only one of them changes its rank, that is, either~$r(v_i,2)=3$ or~$r(\overline{v_i},2)=1$: there is no gain when both change their rank in the time interval~$[2, \alpha  + 2]$.
  This allows us to define a variable assignment~$\beta$ as follows:
  \begin{align*}
  \beta(x_i)\coloneqq \begin{cases}
                      1, & r(v_i, 2) = 3\\
                      0, & r(\overline{v}_i, 2) = 1. \\
                    \end{cases}
  \end{align*}

  $t \in [\alpha + 3, \alpha + d + 3]$:
  First, note that in~$A_t$ we have arcs in both directions between~$v_i$ and the three negative literal vertices of the variable~$x_i$ and arcs in both directions between~$\overline{v}_i$ and the three positive literal vertices.
  So no matter how the ranking of these vertices is in this time interval each such literal-variable pair of vertices incurs a penalty of at least~$2$ in each of the~$d+1$ time steps.
  Since we have six such pairs for every variable, this yields a penalty of at least~$12n(d+1)$.
  As observed above, for each variable~$x_i$, we either have~$r(v_i, \alpha + 2) = 3$ or~$r(\overline{v}_i, \alpha + 2) = 1$.
  If~$r(v_i, \alpha + 2) = 3$, then also~$r(v_i, t) = 3$ for all~$t \in [\alpha + 3, \alpha + d + 3]$ since~$v_i$ cannot change its rank anymore.
  Consequently, every negative literal vertex~$c_{j_z}$ (which is assumed to have rank 1 at time step~$\alpha + 2$) is forced to change its rank to~$2$ or~$3$ in this time interval as otherwise the arc~$(v_i, c_{j_z})$ yields penalty~$3$ (instead of~$2$) in~$d+1$ time steps and the temporal agony would be at least~$12n(d+1) + d+1 > \alpha$.
  Since every negative literal vertex has an outgoing arc to all~$\alpha + 1$ vertices~$b_i$ (which have rank 3), it follows that the negative literal vertices change their rank to~$2$ (and not to rank $3$).
  A similar consequence holds for the positive literal vertices if~$r(\overline{v}_i, \alpha + 2) = 1$, that is, all positive literal vertices of the variable~$x_i$ are forced to change their rank from~$3$ to~$2$ in this time interval.

  $t=\tau$: Since~$q(r,\G,p_l)\le\alpha$ and there is already a penalty of~$\alpha - d$ before the last time step, we know that time step~$\tau$ causes at most~$d$ penalty.
  By the observations above, we know that if a clause~$(c_{j_1} \lor c_{j_2})$ is not satisfied, then the corresponding literal vertices have rank~$2$ at time step~$\tau$ and cause a penalty of~$1$.
  Hence, at most~$d$ clauses are not satisfied by~$\beta$.\qed
\end{proof}

We point out that the lifetime of the constructed temporal digraph in \Cref{thm:three_ranks} can be made constant by using arc weights (that is, reducing to \SEG instead of \USEG).
Next, we show that \SEG with one change point per vertex and three ranks is fixed-parameter tractable with respect to the temporal agony.
Thus, the NP-hardness from \Cref{thm:three_ranks} cannot be extended to constant temporal agony.

\begin{theorem}
	\label{thm:three_ranks_FPT}
	\SEG with $k=3$ and~$\ell = 1$ is solvable in~$2^{O(\alpha)} \cdot L^{O(1)}$ time, where~$L$ is the encoding length of the~input.
\end{theorem}

\begin{proof}
Let~$\G=(V,(A_t)_{t\in[\tau]},(w_t)_{t\in[\tau]})$ be a given temporal weighted digraph in which we want to find a rank 1-segmentation with temporal agony at most~$\alpha$.

The algorithm first identifies vertices which have rankings without any penalty on the incident arcs.
Note that there are at most~$\alpha$ vertices which introduce non-zero penalties.
The key step is thus to efficiently enumerate possible ranking types (but not change points) for all size-$\alpha$ vertex subsets using certain universal set constructions.
For each possible vertex-to-ranking-type-mapping, we then use a~\textsc{Max 2-SAT} encoding to select the best change points for a rank 1-segmentation.
We start with some definitions and observations regarding the temporal agony which reveal a certain structure of the costs which can be encoded in 2-CNF formulas later. 

For a vertex~$v$, we call a mapping~$r_v \colon [\tau] \to [3]$ a \emph{ranking} of~$v$.
Accordingly, for a given rank 1-segmentation~$r\colon V\times [\tau]\to[3]$, we define the ranking of~$v$ as~$r_v\coloneqq r(v, \cdot)$.
Recall the definition of temporal agony:
\[q(r,\G,p_l)\coloneqq \sum_{t\in[\tau]}\sum_{(u,v)\in A_t}w_t(u,v)\cdot \max \{0, r(u,t)-r(v,t) + 1\}.\]

Now, for a vertex~$v$ let~$I_v(t)\coloneqq\sum_{(u,v)\in A_t}w_t(u,v)$ and~$O_v(t)\coloneqq\sum_{(v,u)\in A_t}w_t(v,u)$ denote the sum of weights of all incoming and outgoing arcs of vertex~$v$ at time step~$t$, respectively.
For a ranking~$r_v$ of a vertex~$v$,
the following costs, called \emph{unary costs} of~$v$, are independent of the rankings of the other vertices
\[
U_v(r_v) \coloneqq \sum_{t \in [\tau]} \big( I_v(t) \mathds{1}_{r_v(t) = 1} + O_v(t) \mathds{1}_{r_v(t) = 3} \big).
\]

Note that unary vertex costs already account for the costs of certain arcs.
In particular, for an arc~$(u,v)\in A_t$, the weight $w_t(u,v)$ is already contained in the sum $U_v(r_v)$ if~$r_v(t)=1$ and in $U_u(r_u)$ if $r_u(t)=3$.
Hence, if $r_u(t)=1$, then the cost of~$(u,v)$ is either 0 (that is, $r_v(t)>1$) or contained in~$U_v(r_v)$ (that is, $r_v(t)=1$).
Similarly, the cost of~$(u,v)$ is contained in $U_u(r_u)$ if~$r_v(t)=3$.
In these cases, we say that the arc $(u,v)$ at time step~$t$ is \emph{covered} by~$r_u$ (or by~$r_v$) respectively.
Thus, we can rewrite the temporal agony as
\begin{equation}
\label{eq:temporal-agony}
q(r,\G,p_l) = \sum_{v \in V} U_v(r_v) + \sum_{t \in [\tau]} \sum_{a \in A_t} w_t(a) \mathds{1}_{a \text{ is not covered}}.
\end{equation}
Given two rankings~$r_v$ and~$r_v'$ of a vertex~$v$, we say that~$r_v$ \emph{dominates}~$r_v'$ if~$U_v(r_v) \leq U_v(r_v')$ and every arc which is covered by~$r_v'$ is also covered by~$r_v$.
Clearly, replacing~$r_v'$ by~$r_v$ in some temporal rank assignment does not increase the temporal agony (since both sums above do not increase).
Moreover, we define six \emph{types} of vertex rankings:~$a \to b$ with~$a\neq b \in [3]$, where $a\to b$ denotes rankings in which rank~$a$ is changed to rank~$b$ at some \emph{change point}~$t$ (with $t=1$ and~$t = \tau+1$ being permitted for constant rankings).
Let~$\mathcal{T}$ be the set of types.

As a preprocessing, our algorithm first checks (in polynomial time) for each vertex~$v$ whether there is a ranking~$r_v$ with~$U_v(r_v)=0$ that covers all incident arcs of~$v$.
For this, it is sufficient to consider the~$2\tau$ possible rankings of type~$1 \to 3$ or~$3 \to 1$ since an incident arc cannot be covered by a ranking which assigns rank~2 to~$v$.
If such a ranking exists, then we fix it for this vertex (since this is always optimal).

Next, we show that for every other (non-fixed) vertex~$v$ we can identify two rankings~$P_v$ and~$S_v$ such that every ranking~$r_v$ with~$U_v(r_v)=0$ is dominated by~$P_v$ or~$S_v$.
Consequently, we can later restrict ourselves to solutions~$r$ in which every zero-unary-cost vertex~$v$, which was not already fixed by the preprocessing, either has ranking~$P_v$ or~$S_v$.
To verify the claim, assume that~$v$ was not fixed and that~$r_v$ is a ranking of~$v$ with $U_v(r_v)=0$.
Then, at least one incident arc is not covered by~$r_v$ and since~$r_v$ has zero unary cost, it follows that~$r_v(t) = 2$ for some~$t \in [\tau]$.
Thus, $r_v$ has one of the four types~$1 \to 2, 3 \to 2, 2\to 1$, or~$2 \to 3$.
We define~$P_v$ to be a ranking of type~$a \to 2$ with~$a \in \{1,3\}$ which covers the most arcs while having zero unary cost.
Analogously, we define~$S_v$ as a ranking of type~$2 \to b$ with~$b \in \{1,3\}$.
It is clear that every other ranking of~$v$ with zero unary cost is dominated by~$P_v$ or~$S_v$:
Constant rank~$1$ or~$3$ is not possible since the vertex would have already been fixed in this case.
Constant rank 2 is clearly dominated by~$P_v$.
If~$P_v$ covers at least one arc and has type~$1 \to 2$, then the first incident arc of~$v$ is an outgoing arc.
Thus, non-constant rankings of type~$3 \to 2$ with zero unary cost cannot cover any arc and consequently are dominated by~$P_v$.
Moreover, by definition, every ranking of type~$1 \to 2$ with zero unary cost is dominated by~$P_v$.
Similar arguments hold if~$P_v$ has type~$3 \to 2$ and also show that~$S_v$ dominates all rankings with type~$2 \to 1$ or~$2 \to 3$.
Additionally, we remark that~$P_v$ and~$S_v$ cover arcs at disjoint time steps since otherwise there is a ranking of type~$1 \to 3$ or~$3 \to 1$ which covers all arcs and has zero unary cost.
To sum up, we can assume that in a solution every zero-unary-cost vertex~$v$ (which was not fixed by the preprocessing) has ranking~$P_v$ or~$S_v$.
Next, we discuss how the algorithm determines possible ranking types for vertices with positive unary cost.
Afterwards, we show how to choose final rankings by solving~\textsc{Max 2-SAT} instances.

Since every solution yields at most~$\alpha$ vertices with positive unary cost, we need to consider all mappings~$h \colon X \to \mathcal{T}$ from every vertex subset~$X \subseteq V$ of size~$|X| \leq \alpha$ to the set of ranking types~$\mathcal{T}$.
To do this, we use \emph{universal sets}.
A set~$T \subseteq \{0, 1,\dots, b - 1\}^n$ is an~$(n, k, b)$-universal set if, for every index set~$S \subseteq \{1,\dots, n\} $ with~$|S| = k$, the projection of~$T$ on~$S$ contains all possible~$b^k$ configurations~\cite{UllrichV22}[Section 2.1].
With~$k = \alpha$ and~$b = |\mathcal T|=6$, an $(n,\alpha,6)$-universal set can be constructed in~$2^{O(\alpha)} n^{O(1)}$ time~\cite{NaorSS95, UllrichV22}.
This means we can construct in~$2^{O(\alpha)} n^{O(1)}$ time a set~$\mathcal{H}$ of mappings~$h \colon V \to \mathcal{T}$ with~$|\mathcal{H}| \in 2^{O(\alpha)} n^{O(1)}$ such that our instance has a solution if and only if it has a solution and there exists an~$h \in \mathcal{H}$ such that every non-fixed vertex~$v$ has ranking~$P_v,S_v$ or a ranking of type~$h(v)$ in this solution.

With the set~$\mathcal{H}$ in hand, the idea is the following: For every~$h \in \mathcal{H}$, a vertex can choose between rankings~$P_v$, $S_v$ and certain \emph{reasonable} rankings of type~$h(v)$.
To define a ranking being reasonable, let~$f_I, f_O, \ell_I$ and~$\ell_O$ be the time steps of the first/last incoming/outgoing arc of~$v$ (these time steps exist as~$v$ was not fixed yet).
Then, a ranking~$r_v$ of type~$h(v)$ is reasonable if its change point~$c_{r_v}\in[\tau]$ satisfies the following condition depending on~$h(v)$:
\begin{center}
    \begin{tabular}{c|c}
        $h(v)$ & Condition \\
        \hline
        $1\to 2$ & $f_I \leq c_{r_v}$ \\
        $1\to 3$ & $f_I \leq c_{r_v} \leq \ell_O + 1$ \\
        $3\to 2$ & $f_O \leq c_{r_v}$ \\
        $2\to 1$ & $c_{r_v} \leq \ell_I + 1$ \\
        $3\to 1$ & $f_O \leq c_{r_v} \leq \ell_I + 1$ \\
        $2\to 3$ & $c_{r_v} \leq \ell_O + 1$ \\
    \end{tabular}
\end{center}

Intuitively, the above conditions in the table ensure that a vertex does not change its rank at a time step where it is definitely not necessary.
For example, the first row states that a vertex does not need to change its rank from 1 to 2 before the first incoming arc.
Formally, the conditions imply that every ranking of type~$h(v)$ is dominated by a reasonable ranking of the same type.
Therefore we can restrict our search for a solution to reasonable rankings.
We now show that this fact allows us to choose the final ranking by solving~\textsc{Max 2-SAT}.

From now on, let~$R_{v,1}, \dots, R_{v,m_v}$ be the reasonable rankings of~$v$ with type~$h(v)$ in increasing change point order (computable in polynomial time).
The choice of rankings together with the costs can be encoded in a polynomial-size (weighted) 2-CNF formula~$\phi_h$.
We construct the formula~$\phi_h$ such that selecting rankings from the lists~$\mathcal R_v\coloneqq (P_v, R_{v,1}, \dots, R_{v,m_v}, S_v)$ with temporal agony at most~$\alpha$ is equivalent to finding an assignment whose unsatisfied clauses have total weight at most~$\alpha$.
We first construct a weighted 2-CNF formula and afterwards discuss how to get rid of the clause weights to apply an algorithm by~\citet{RazgonO09}.

For every non-fixed vertex~$v$, let~$\mathcal{R}_v = (r_{v,0}, \dots, r_{v,m_v+1})$ be the possible ranking choices.
We introduce Boolean variables~$x_{v,0}, \dots, x_{v, m_v + 2}$ where an assignment with $x_{v,0} = \dots = x_{v,s} = 1$ and~$x_{v,s+1} = \dots = x_{v,m_v+2} = 0$ corresponds to~$r_{v,s}$ being selected as the ranking of~$v$.
We enforce~$x_{v,0} = 1$ and~$x_{v,m_v+2} = 0$ by adding the unit clauses~$(x_{v,0})$ and~$(\lnot x_{v,m_v+2})$ with weight~$\alpha + 1$.
Moreover, for every~$i \in [0,m_v + 1]$, we add the clause~$(x_{v,i+1} \rightarrow x_{v,i})$ with weight~$\alpha + 1$.
These clauses enforce that some starting sequence of the variables~$x_{v,0}, \dots, x_{v, m_v + 2}$ is set to~1 while the remaining variables are set to~0.
For every ranking~$r_{v,i}$ from~$\mathcal{R}_v$ with~$U_v(r_{v,i})>0$, we add the clause~$(x_{v,i} \rightarrow x_{v,i+1})$ with weight~$U_v(r_{v,i})$.
This encodes the unary cost of the selected ranking.

By Equation~\eqref{eq:temporal-agony}, it remains to add the costs of uncovered arcs.
The following claim is helpful for the encoding of these costs.

 \begin{claim}
  Let~$a\in A_t$ be an arc which is incident to~$v$.
  Then the index set~$\{i \in [0, m_v+1] \mid r_{v,i} \text{ covers } a\}$ forms a (possibly empty) subinterval of~$[0, m_v+1]$.
 \end{claim}
 
  To prove this claim, let~$r_{v,i}$ with~$i \in [m_v]$ be a ranking which covers arc~$a$.
  Since the rankings~$r_{v,1}, \dots, r_{v,m_v}$ have the same type and are sorted in increasing change point order, it is clear that~$r_{v,1}, \dots, r_{v,i-1}$ or~$r_{v,i+1}, \dots, r_{v,m_v}$ also cover the arc~$a$.
  Moreover, if~$r_{v,i}$ and~$r_{v,0} = P_v$ cover~$a$, then~$r_{v,1}, \dots, r_{v,i-1}$ cover the arc~$a$:
  Let~$t_P$ be the change point of~$P_v$ where it changes its rank from~$x \in \{1,3\}$ to~$2$.
  If~$x=1$, then~$v$ has no incoming arcs until time step~$t_P$ since the ranking~$P_v$ has zero unary cost.
  If~$r_{v,i}$ has type~$y \to 1$ with~$y \in [3]$, then clearly~$r_{v,1}, \dots, r_{v,i-1}$ cover arc~$a$ at time step~$t$ because their change points lie before the change point of~$r_{v,i}$.
  If~$r_{v,i}$ has type~$1 \to y$ with~$y \in [3]$, then the change point of every ranking~$r_{v,1}, \dots, r_{v,i-1}$ is after time step~$t$ since they are reasonable and therefore change their rank only once the first incoming arc appears.
  Thus, the rankings~$r_{v,1}, \dots, r_{v,i-1}$ cover~$a$ at time step~$t$.
  Similar arguments hold for~$x=3$ and~$S_v$ with~$r_{v,i}, \ldots,  r_{v,m_v+1}$.
  Finally, recall that~$P_v$ and~$S_v$ cover arcs at disjoint time steps.
  Hence, the index set from the claim forms a subinterval of~$[0, m_v+1]$.

Now, assume that~$a=(u,v) \in A_t$ is not already covered by a ranking of a fixed vertex and let~$[L_u, R_u] \subseteq [0, m_u+1]$ and~$[L_v,R_v] \subseteq [0, m_v+1]$ be the respective intervals from Claim~1 above.
If both intervals are non-empty, then we add the following four clauses with weight~$w_t(a)$ to our formula:
\begin{align*}
 &(\neg x_{u,L_u} \rightarrow x_{v,L_v}), \qquad
(\neg x_{u,L_u} \rightarrow \neg x_{v,R_v+1}), \\
&(x_{u,R_u+1}\rightarrow x_{v,L_v}), \qquad
(x_{u,R_u+1}\rightarrow\neg x_{v,R_v+1}).
\end{align*}
Recall that~$x_{u,L_u} = 1$ and~$x_{u,R_{u+1}} = 0$ means that the selected ranking is a ranking from~$(r_{u,{L_u}}, \dots, r_{u,{R_u}})$ and therefore covers the arc.
The same holds for vertex~$v$ and therefore all four clauses are satisfied if and only if the arc is covered by the selected ranking of~$u$ or~$v$.
Otherwise, if none of the selected rankings covers the arc, then exactly one of the clauses is violated.
If exactly one of the intervals is empty (w.l.o.g. the interval for~$u$), then we add the unit clauses $(x_{v,L_v})$ and $(\neg x_{v,R_v+1})$ with weight~$w_t(a)$.
If both of them are empty, then we add an empty clause with weight~$w_t(a)$.

Finally, we get rid of the weights by replacing a clause of weight~$w$ by $\min \{w, \alpha + 1\}$ copies of the clause.
This allows us to apply the following result from~\citet{RazgonO09} to conclude the proof:
Given a 2-CNF formula with~$m$ clauses (where repeated occurrences are permitted), we can check in~$2^{O(\alpha)} \cdot m^{3}$ time whether there exists an assignment that violates at most~$\alpha$ clauses.

Fixing the vertices in the beginning and identifying the rankings~$P_v$ and~$S_v$ for each vertex can clearly be done in polynomial time.
Afterwards, the set~$\mathcal{H}$ is constructed in~$2^{O(\alpha)} n^{O(1)}$ time.
Finally, for each~$h \in \mathcal{H}$ the algorithm constructs a~\textsc{Max 2-SAT}-instance with polynomially many different clauses of which each appears at most~$\alpha + 1$ many times in the final formula.
Thus, using~$|\mathcal{H}| \in 2^{O(\alpha)} n^{O(1)}$ and the~\textsc{Max 2-SAT} algorithm from~\citet{RazgonO09}, we can upper-bound the running time with~$2^{O(\alpha)} \cdot L^{O(1)}$ where~$L$ is the encoding length of the~input.
\qed
\end{proof}

In the next section, we show that the above algorithm cannot be extended to four ranks.
The main problem is that it is not clear how to choose the zero-unary cost rankings~$P_v$ and~$S_v$ for a non-fixed vertex~$v$.
For example, when considering four ranks, if the first incident arc of~$v$ is outgoing, then~$P_v$ could be a ranking of type~$1 \to 2$ or~$1\to 3$.
This leaves enough choice to obtain a more restricted hardness result for four ranks in the next section.

\section{Four Ranks}
\label{sec:four_ranks}

In this section, we show that \textsc{(Unweighted)} \SEG with $k=4$ ranks is strictly harder than~$k=3$ ranks as the former one is already NP-hard for zero agony and a constant lifetime.
In fact, this even holds if the underlying digraph is a DAG.
Note that the constructed instance in the reduction of \Cref{thm:three_ranks} for~$k=3$ has directed cycles in the underlying digraph.
It is unclear whether the case~$k=3$ is also NP-hard if the underlying digraph is a DAG.

\begin{restatable}{theorem}{fourRanks}
	\label{thm:four_ranks}
	\USEG with $k=4$, $\ell =1$, $\alpha = 0$ and $\tau=11$ is NP-hard even if the underlying digraph is a DAG.
\end{restatable}

\begin{proof}
 We give a polynomial-time reduction from \textsc{Monotone 3-SAT} where the input is a~$3$-CNF formula~$\Phi$ consisting of monotone clauses \cite{DarmannD21}.
 Here, a clause is called monotone if it contains only negative or only positive literals.
 Suppose~$\Phi$ has the variable set~$\{x_1, \dots, x_n\}$ and the clause set~$\{C_1, \dots, C_m\}$ where the literals in every clause are enumerated, that is, $C_j = (\ell_{j_1} \lor \ell_{j_2} \lor \ell_{j_3})$.

 \emph{Construction:} We construct a temporal digraph~$\G=(V,(A_t)_{t\in [11]})$ with $\tau=11$ time steps (see \Cref{fig:four_ranks_main}).
 The vertex set~$V$ is defined as follows:
 \begin{itemize}
  \item We add four vertices~$d_1,d_2,d_3$ and $d_4$, which will be ``anchors'' in the ranks $1,2,3$ and~$4$ respectively.
  \item For each variable~$x_i$, we add one vertex~$v_i$ and refer to it as a \emph{variable vertex}.
  \item For each literal~$\ell_{j_i}$, we add one vertex~$\ell_{j_i}$ to~$V$ and refer to it as a \emph{(negative/positive) literal vertex}.
 \end{itemize}

 It remains to define the arc sets.
 First, we add the arcs~$(d_1,d_2), (d_2, d_3)$ and~$(d_3,d_4)$ to~$A_1$ and~$A_{11}$.
 This ensures that in any rank 1-segmentation with agony~$0$ the vertices~$d_1,d_2,d_3$ and $d_4$ have rank~$1,2,3$ and~$4$ respectively in every time step.
 For every literal~$\ell_{j_i}$, we add an outgoing arc to its corresponding variable as follows:
 If it is a negative literal, then we add the arc to~$A_1$ and if it is positive to~$A_7$.
 Moreover, we add the arcs~$(d_1,v_i)$ and~$(v_i, d_3)$ for each variable~$x_i$ to~$A_6$.
 Thus, the variable vertex~$v_i$ needs to have rank~$2$ at time step~6.
 This yields the binary choice for each variable to choose whether its vertex switches the rank before or after time step~6.
 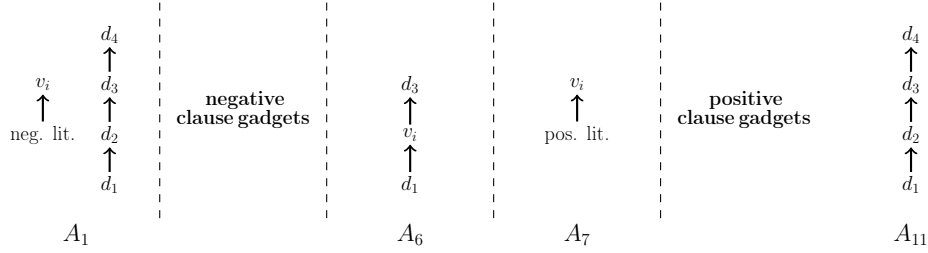
\begin{figure}[t]
	\centering
	\begin{tikzpicture}[scale = 0.4412, transform shape,
		V/.style = {circle, draw, fill=black}
		]
		\node (negLit) at (-11.5, 1.5) {\LARGE neg. lit.};
		\node (vi1) at (-11.5, 3) {\LARGE$v_i$};
		\node (A1) at (-9.5,0) {\LARGE$d_1$};
		\node (B1) at (-9.5, 1.5) {\LARGE$d_2$};
		\node (C1) at (-9.5,3) {\LARGE$d_3$};
		\node (D1) at (-9.5,4.5) {\LARGE$d_4$};
		\node (T1) at (-10.5, -1.5) {\huge $A_1$};
		\draw[->,thick] (negLit) edge (vi1);
		\draw[->,thick] (A1) edge (B1);
		\draw[->,thick] (B1) edge (C1);
		\draw[->,thick] (C1) edge (D1);

		\node (vi6) at (-0.5, 1.5) {\LARGE$v_i$};
		\node (C6) at (-0.5,3) {\LARGE$d_3$};
		\node(D6) at (-0.5,0) {\LARGE$d_1$};
		\node (T6) at (-0.5, -1.5) {\huge $A_6$};
		\draw[->,thick] (vi6) edge (C6);
		\draw[->,thick] (D6) edge (vi6);

		\node (posLit) at (4.5, 1.5) {\LARGE pos. lit.};
		\node (vi7) at (4.5, 3) {\LARGE$v_i$};
		\node (T7) at (4.5, -1.5) {\huge $A_7$};
		\draw[->,thick] (posLit) edge (vi7);

		\draw[dashed] (-8, -1) -- (-8, 5.5);\node (negClauses) at (-5.5, 2.25)
        [text width=4cm, align=center]
        {\LARGE\textbf{ negative\\ clause gadgets}};
		\draw[dashed] (-3, -1) -- (-3, 5.5);
		\draw[dashed] (2, -1) -- (2, 5.5);
		\draw[dashed] (7, -1) -- (7, 5.5);\node (posClauses) at (9.5, 2.25)
        [text width=4cm, align=center]
        {\LARGE\textbf{ positive\\ clause gadgets}};
		\node (A11) at (14.5,0) {\LARGE$d_1$};
		\node (B11) at (14.5,1.5) {\LARGE$d_2$};
		\node (C11) at (14.5,3) {\LARGE$d_3$};
		\node (D11) at (14.5,4.5) {\LARGE$d_4$};
		\node (T11) at (14.5, -1.5) {\huge $A_{11}$};
		\draw[->,thick] (A11) edge (B11);
		\draw[->,thick] (B11) edge (C11);
		\draw[->,thick] (C11) edge (D11);

	\end{tikzpicture}
	\caption{A sketch of the reduction from \textsc{Monotone 3-SAT} in \Cref{thm:four_ranks}.
	The constructed temporal digraph consists of eleven time steps and time step~6 forces vertex~$v_i$ into rank~$2$.
	The binary choice of the corresponding Boolean variable is encoded by the choice of whether~$v_i$ has rank~$2$ in the time interval~$[1,6]$ or in the interval~$[6,11]$.
	The negative clause gadgets lie in the time interval~$[2,5]$ and the positive clause gadgets in the time interval~$[8,11]$.}
	\label{fig:four_ranks_main}
\end{figure}

 To encode the negative clauses in the time interval~$\{2,3,4,5\}$ and the positive clauses in the time interval~$\{8,9,10,11\}$, we construct a clause gadget as follows (see \Cref{fig:four_ranks_clause_gadgets}):
 For a clause~$C_j = (\ell_{j_1} \lor \ell_{j_2} \lor \ell_{j_3})$, %
we say that we ``plant'' the clause gadget at time step~$i$ by adding the following arcs:
 \begin{itemize}
  \item $A_i$: add~$(d_1,\ell_{j_1})$,
  \item $A_{i+1}$: add~$(\ell_{j_1}, \ell_{j_2})$,
  \item $A_{i+2}$: add~$(\ell_{j_1}, \ell_{j_3})$ and~$(\ell_{j_3}, d_4)$,
  \item $A_{i+3}$: add~$(\ell_{j_2}, \ell_{j_3})$.
 \end{itemize}

 As indicated above, if the clause is negative, then the respective gadget is planted at time step~$2$ and otherwise at time step~$8$.
 This concludes our construction. %
 Note that the underlying digraph is indeed a DAG (see \Cref{fig:four_ranks_clause_gadgets} for a sketch).
\begin{figure}[t]
	\centering
	\begin{tikzpicture}[scale = 0.57, transform shape,
		V/.style = {circle, draw, fill=black}
		]
		\node (A1) at (-10,0.75) {\LARGE$d_1$};
		\node (B1) at (-10, 2.25) {\LARGE$\ell_{j_1}$};
		\node (T1) at (-10, -1) {\huge $A_i$};
		\draw[->,thick] (A1) edge (B1);

		\node (A2) at (-6,0.75) {\LARGE$\ell_{j_1}$};
		\node(B2) at (-6,2.25) {\LARGE$\ell_{j_2}$};
		\node (T2) at (-6, -1) {\huge $A_{i+1}$};
		\draw[->,thick] (A2) edge (B2);

		\node (A3) at (-2,0.75) {\LARGE$\ell_{j_1}$};
		\node(B3) at (-2,2.25) {\LARGE$\ell_{j_3}$};
		\node(C3) at (-2,3.75) {\LARGE$d_4$};
		\node (T3) at (-2, -1) {\huge $A_{i+2}$};
		\draw[->,thick] (A3) edge (B3);
		\draw[->,thick] (B3) edge (C3);

		\node (A4) at (2,0.75) {\LARGE$\ell_{j_2}$};
		\node(B4) at (2,2.25) {\LARGE$\ell_{j_3}$};
		\node (T4) at (2, -1) {\huge $A_{i+3}$};
		\draw[->,thick] (A4) edge (B4);

		\draw[dashed] (-8, 0) -- (-8, 4);
		\draw[dashed] (-4, 0) -- (-4, 4);
		\draw[dashed] (0, 0) -- (0, 4);
		\draw[thick] (5.5, -1) -- (5.5, 5);

		\node (d1) at (9,-1) {\large$d_1$};
		\node(l1) at (9,0) {\large$\ell_{j_1}$};
		\node(l2) at (9,1) {\large$\ell_{j_2}$};
		\node (l3) at (9,2) {\large $\ell_{j_3}$};
		\node(vi) at (9,3) {\large$v_i$};
		\node(d2) at (10.5,1.5) {\large$d_2$};
		\node (d3) at (9,4) {\large $d_3$};
		\node(d4) at (9,5) {\large$d_4$};
		\draw[->,thick] (d1) edge (l1);
		\draw[->,thick] (l1) edge (l2);
		\draw[->,thick] (l2) edge (l3);
		\draw[->,thick] (l3) edge (vi);
		\draw[->,thick] (vi) edge (d3);
		\draw[->,thick] (d1) edge (d2);
		\draw[->,thick] (d2) edge (d3);
		\draw[->,thick] (d3) edge (d4);

	\end{tikzpicture}
	\caption{On the left is a sketch of a clause gadget from the reduction in \Cref{thm:four_ranks}.
	Recall that~$d_1$ and~$d_4$ are fixed at rank~$1$ and~$4$ respectively.
	If~$\ell_{j_1}, \ell_{j_2}$ and~$\ell_{j_3}$ have rank~$1$ at time step~$i-1$, then this gadget causes agony.
	If at least one of the three literal vertices can choose its rank freely from~$\{1,2,3\}$ at time step~$i-1$, then agony can be avoided.
	On the right is a schematic drawing of the underlying DAG of~$\mathcal G$, with some arcs omitted for clarity}
	\label{fig:four_ranks_clause_gadgets}
\end{figure}
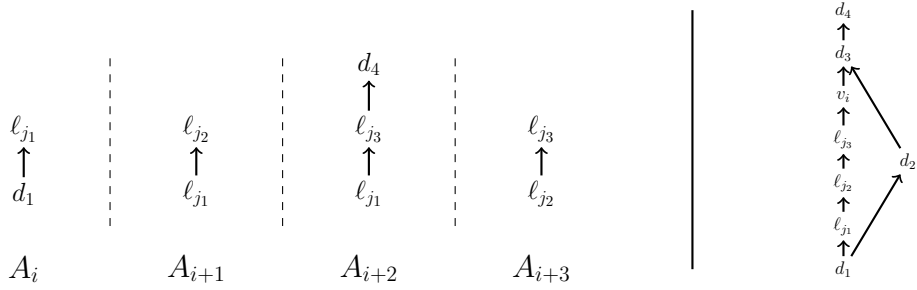

 \emph{Correctness:} %
 ``$\Rightarrow$'' Let~$\beta$ be a satisfying assignment for~$\Phi$.
 Consider the following rank 1-segmentation~$r$:
 \begin{itemize}
  \item For all~$t\in [11]$, we define~$r(d_1,t) \coloneqq 1, r(d_2,t) \coloneqq 2, r(d_3,t) \coloneqq 3$ and~$r(d_4,t) \coloneqq 4$.
  \item If~$\beta(x_i) = 1$, then we define
  \begin{align*}
  r(v_i,t)\coloneqq \begin{cases}
                      2, &t \leq 6\\
                      4, &t > 6
                    \end{cases}
  \end{align*}
  and if~$\beta(x_i) = 0$, then we define
  \begin{align*}
  r(v_i,t)\coloneqq \begin{cases}
                      4, &t < 6\\
                      2, &t \geq 6.
                    \end{cases}
  \end{align*}
 \end{itemize}

 It remains to define the ranks of the literal vertices.
 Let~$\ell_{j_i}$ be a negative literal of variable~$x_h$.
 If it is set to false by~$\beta$ (that is, $\beta(x_h)=1$), then~$r(v_h,1)=2$ forces~$r(\ell_{j_i}, 1) = 1$.
If it is set to true by~$\beta$ (that is, $\beta(x_h)=0$), then~$r(v_h,1)=4$ and hence $r(\ell_{j_i}, 1)$ can be freely chosen from $\{1,2,3\}$.
The same observation holds for positive literal vertices at time step~$7$.

Now, we claim that agony can be avoided in a clause gadget planted at time step~$i>1$ if at least one of the literal vertices can choose its rank from~$\{1,2,3\}$ at time step~$i-1$.
  The proof is via case distinction:%
  \begin{itemize}
   \item If~$\ell_{j_1}$ can choose, then we define
   \begin{align*}
    r(\ell_{j_1},t)&\coloneqq
    \begin{cases}
    2, & t \leq i\\
    1, & t > i
    \end{cases}, \quad
    r(\ell_{j_2},t)\coloneqq
    \begin{cases}
    1, & t \leq i\\
    2, & t > i
    \end{cases}, \\
    r(\ell_{j_3},t)&\coloneqq
    \begin{cases}
    1, & t \leq i+1\\
    3, & t > i+1
    \end{cases}.
   \end{align*}
   \item If~$\ell_{j_2}$ can choose, then we define
   \begin{align*}
    r(\ell_{j_1},t)&\coloneqq
    \begin{cases}
    1, & t \leq i-1\\
    2, & t > i-1
    \end{cases}, \quad
    r(\ell_{j_2},t)\coloneqq
    \begin{cases}
    3, & t \leq i+2\\
    2, & t > i+2
    \end{cases}, \\
    r(\ell_{j_3},t)&\coloneqq
    \begin{cases}
    1, & t \leq i+1\\
    3, & t > i+1
    \end{cases}.
   \end{align*}
   \item If~$\ell_{j_3}$ can choose, then we define
   \begin{align*}
    r(\ell_{j_1},t)&\coloneqq
    \begin{cases}
    1, & t \leq i-1\\
    2, & t > i-1
    \end{cases}, \quad
    r(\ell_{j_2},t)\coloneqq
    \begin{cases}
    1, & t \leq i\\
    3, & t > i
    \end{cases}, \\
    r(\ell_{j_3},t) &\coloneqq
    \begin{cases}
    3, & t \leq i+2\\
    4, & t > i+2
    \end{cases}.
   \end{align*}
  \end{itemize}
  It is easily verified that none of the three cases above produces agony (by construction of a clause gadget).

 Since~$\beta$ is a satisfying assignment, we know that in each clause at least one literal is set to true. Hence, for every clause gadget there is a literal with a choice before the gadget and consequently there is a rank 1-segmentation~$r$ with~$q(r,\G,p_l)=0$.

 ``$\Leftarrow$'' Let~$r$ be a rank 1-segmentation with~$q(r,\G,p_l) = 0$.
 Clearly, $r(d_i,t)=i$ for all~$i\in[4]$ and~$t\in[11]$.
 This implies~$r(v_i, 6)= 2$ for each variable vertex~$v_i$ since~$(d_1, v_i), (v_i,d_3) \in A_6$.
 Since~$v_i$ is a sink in every other time step than~$6$, we can assume either~$r(v_i, 5)= 4$ or~$r(v_i, 7)= 4$.
 This allows us to define a variable assignment~$\beta$ as follows:
 \begin{align*}
  \beta(x_i)\coloneqq
  \begin{cases}
  1, & r(v_i, 5) = 2\\
  0, & r(v_i, 7) = 2\\
  \end{cases}.
 \end{align*}
 It remains to show that~$\beta$ is a satisfying assignment of~$\Phi$.
 Suppose it is not and there is a negative clause~$(\ell_{j_1} \lor \ell_{j_2} \lor \ell_{j_3})$ in which all literals are set to false by~$\beta$.
 Then all corresponding variable vertices have rank~$2$ at time step~$1$, which implies~$r(\ell_{j_1},1) = r(\ell_{j_2},1) = r(\ell_{j_3},1) = 1$.
 By construction, this produces agony in the clause gadget which is a contradiction.
  This can be seen as follows: Since~$d_1$ has rank~$1$ it is clear that~$\ell_{j_1}$ has to switch from rank~$1$ to rank at least~$2$ at time step~2.
  Subsequently, $\ell_{j_2}$ has to switch from~$1$ to at least~$3$ at time step~3 and $\ell_{j_3}$ has to switch from~$1$ to~$3$ at time step~$4$ since~$d_4$ has rank~$4$.
  Thus, at time step~5, the arc~$(\ell_{j_2}, \ell_{j_3})$ produces agony since the rank of~$\ell_{j_2}$ is at least the rank of~$\ell_{j_3}$.

  For positive clauses, the argument is the same with time step~$7$ instead of~$1$.\qed
\end{proof}

\section{Conclusion}

We studied the parameterized complexity of \SEG with a focus on the parameter number of ranks.
In particular, we showed that there are jumps in the complexity when going from two to three and from three to four ranks.
On an intuitive level it seems like the restriction to~$k$ ranks and zero agony corresponds to $(k-1)$\textsc{-SAT}.
With zero agony, two ranks are straightforward, three ranks can be encoded as a \textsc{2-SAT}-instance and four ranks are NP-hard by a reduction from a \textsc{3-SAT} variant.
Accordingly and as undermined by the FPT result for~$\alpha$ when~$k=3$, the problem of minimizing the temporal agony for~$k$ ranks roughly corresponds to \textsc{Max} $(k-1)$\textsc{-SAT}.
Interesting research directions for the future are the following:
\begin{compactitem}
 \item How does the complexity of \SEG with $k=3$ change if we allow more than one change point per vertex (that is, ~$\ell > 1$)?
 \item Is the problem in FPT with respect to~$n$? (This is open even for $k=3$.)
 \item Are there similar results regarding the number of ranks if we use a global budget for the total number of vertex rank changes?
\end{compactitem}

\subsection*{Statement on the Use of Generative AI}
The authors used GPT‑6 Astra to prove \Cref{thm:three_ranks_FPT}. The authors reviewed and edited the content as needed and take full responsibility for the content of the~article.

\bibliographystyle{plainnat}

\bibliography{ref}

\end{document}